\documentclass{article}

\usepackage{amsmath,amsthm,amssymb}

\allowdisplaybreaks[3]

\numberwithin{equation}{section}

\newtheorem{theorem}{Theorem}[section]
\newtheorem{corollary}[theorem]{Corollary}
\newtheorem{lemma}[theorem]{Lemma}
\newtheorem{proposition}[theorem]{Proposition}
\theoremstyle{definition}

\theoremstyle{remark}

\newcommand{\ab}{\allowbreak}

\title{An approximative approach to construction of the Glauber dynamics in continuum%
\thanks{The financial support of DFG through the SFB 701
(Bielefeld University) and German-Ukrainian Project 436 UKR 113/94
and 436 UKR 113/97, and RFFI grant 08-01-00105a is gratefully
acknowledged.}}

\author{Dmitri Finkelshtein\thanks{Institute of Mathematics,
National Academy of Sciences of Ukraine, Kyiv, Ukraine ({\tt fdl@\ab
imath.\ab kiev.\ab ua}).} \and Yuri Kondratiev\thanks{Fakult\"{a}t
f\"{u}r Mathematik, Universit\"{a}t Bielefeld, 33615 Bielefeld,
Germany ({\tt kondrat@\ab math.\ab uni-bielefeld.\ab de})} \and
Oleksandr Kutoviy\thanks{Fakult\"{a}t f\"{u}r Mathematik,
Universit\"{a}t Bielefeld, 33615 Bielefeld, Germany ({\tt
kutoviy@\ab math.\ab uni-bielefeld.\ab de}).} \and Elena
Zhizhina\thanks{Institute for Information Transmission Problems,
Moscow, Russia ({\tt ejj@\ab iitp.\ab ru}).}}

\begin{document}

\maketitle

\begin{abstract}
We develop a new approach for the construction of the Glauber
dynamics in continuum. Existence of the corresponding strongly
continuous contraction semigroup in a proper Banach space is shown.
Additionally we present the finite- and infinite-volume
approximations of the semigroup by families of bounded linear
operators.
\end{abstract}



\section{Introduction}
The Glauber type stochastic dynamics in continuum are
birth-and-death Markov processes on configuration spaces with the
given reversible states which are grand canonical Gibbs measures.
The corresponding Markov generators are related with the (non-local)
Dirichlet forms for the considered Gibbs measures. The latter fact
gives a standard way to construct properly associated stationary
Markov processes. These processes preserve the initial Gibbs state
in the time evolution; they are called the equilibrium Glauber
dynamics, see, e.g., \cite{KoLy}, \cite{KoLyRo}, \cite{KoMiZh},
\cite{FiKoLy}. Note that, in applications, the time evolution of
initial state is the subject of the primary interest. In what
follows, we will try to understand the considered stochastic
dynamics as the evolution of initial distributions for the system.
Actually, the Markov process itself gives a general technical
equipment to study this problem. Let us stress that the transition
from the micro-state evolution corresponding to the given initial
configuration  to the macro-state dynamics is the well developed
concept in the theory of infinite particle systems. This point of
view appeared initially in the framework of the Hamiltonian dynamics
of classical gases, see, e.g., \cite{DoSiSu}.

The study of the non-equilibrium Glauber dynamics needs construction
of the time evolution for a wider class of initial measures. The
lack of the general Markov processes techniques for the considered
systems makes necessary to develop alternative approaches to study
the state evolutions in the Glauber dynamics. The approach realized
in \cite{KoKutMi}, \cite{KoKutZh} is probably the only known at the
present time. The description of the time evolutions for measures on
configuration spaces in terms of an infinite system of evolutional
equations for the corresponding correlation functions was used
there. The latter system is a Glauber evolution's analog of the
famous BBGKY-hierarchy for the Hamiltonian dynamics.

Here we develop another approach to the Glauber dynamics in
continuum. This constructive approach was inspired by working up a
new algorithm for detection problems in image processing. Object
detection, or detecting a configuration of objects from a digital
image, is a crucial step in many applications. In paper \cite{DMZ},
a new stochastic algorithm to solve object detection problems has
been proposed.

The algorithm is based on a continuous time stochastic evolution of
macro-objects in a large (but finite) volume in continuum. It was
considered a model of possibly partially overlapping discs. Each
disc in the final configuration is associated with a given object in
the image. The evolution under consideration is a birth-and-death
equilibrium dynamics on the configuration space of discs with a
given stationary Gibbs measure. In this scheme, the intensity of
birth is a constant, whereas intensities of death depend on the
energy function and the present configuration. This choice of rates
has been made to optimize the convergence speed. Indeed, the volume
of the space for birth is much larger than the number of discs in
the configuration. To apply the continuous time dynamics to
simulation process we have to construct a discretization of this
process in time. The resulting discrete time process is a
non-homogeneous Markov chain with transition probabilities depending
on the energy function and the discretization step. The main point
of our approach is that each step in proposed algorithm concerns the
whole configuration, so that it is so-called multiple
birth-and-death algorithm.

In this paper, we introduce and study the analogous discretization
for infinite volume birth and death dynamics of the Glauber type,
and prove the convergence to the continuous time process as the step
of discretization tends to zero. Furthermore, we use the
discretization to construct a non-equilibrium dynamics.

\section{Description of model}

\subsection{General facts and notations}
Let ${\mathcal B}({{\mathbb R}^d})$ be the family of all Borel sets
in ${{\mathbb R}^d}$, $d\geq 1$. ${\mathcal B}_{\mathrm{b}}
({{\mathbb R}^d})$ denotes the system of all bounded sets in
${\mathcal B}({{\mathbb R}^d})$.

We define the space of the $n$-point configurations in
$Y\in{\mathcal B}({{\mathbb R}^d})$ as
\[
\Gamma ^{(n)}_Y:=\bigl\{  \eta \subset Y \bigm| |\eta |=n\bigr\}
,\quad n\in {\mathbb N},
\]
where $|\cdot|$ mean the cardinality of a finite set. We put also
$\Gamma ^{(0)}_Y:=\{\emptyset\}$. As a set, $\Gamma ^{(n)}_Y$ is
equivalent to the symmetrization of
\[
\widetilde{Y^n} = \bigl\{ (x_1,\ldots ,x_n)\in Y^n \bigm| x_k\neq
x_l \text{ if } k\neq l\bigr\} .
\]
Hence, one can introduce the corresponding Borel $\sigma $-algebra,
which we denote by ${\mathcal B}(\Gamma ^{(n)}_Y)$. The space of
finite configurations in $Y\in{\mathcal B}({{\mathbb R}^d})$ defined
as
\[
\Gamma_{0,Y}:=\bigsqcup_{n\in {\mathbb N}_0}\Gamma ^{(n)}_Y
\]
is equipped with the topology of disjoint unions. Therefore, one can
introduce the corresponding Borel $\sigma $-algebra ${\mathcal B}
(\Gamma_{0,Y})$. In the case of $Y={{\mathbb R}^d}$ we will omit the
index $Y$ in the notation, namely, $\Gamma_0:=\Gamma_{0,{{\mathbb
R}^d}}$, $\Gamma ^{(n)}:=\Gamma ^{(n)}_{{{\mathbb R}^d}}$.

The restriction of the Lebesgue product measure $(dx)^n$ to
$\bigl(\Gamma ^{(n)}, {\mathcal B}(\Gamma ^{(n)})\bigr)$ we denote
by $m^{(n)}$. We set $m^{(0)}:=\delta_{\{\emptyset\}}$. Let
$\kappa>0$ be fixed. The Lebesgue--Poisson measure $\lambda_\kappa$
on $\Gamma_0$ is defined as
\[
\lambda_\kappa :=\sum_{n=0}^\infty \frac {\kappa^n}{n!}m ^{(n)}.
\]
For any $\Lambda \in{\mathcal B}_{\mathrm{b}}({{\mathbb R}^d})$ the
restriction of $\lambda_\kappa$ to $\Gamma_\Lambda
:=\Gamma_{0,\Lambda }$ will be also denoted by $\lambda_\kappa $. We
denote also $\lambda=\lambda_1$.

The configuration space over space ${{\mathbb R}^d}$ consists of all
locally finite subsets (configurations) of ${{\mathbb R}^d}$,
namely,
\begin{equation}\label{confspace}
\Gamma  =\Gamma_{{\mathbb R}^d} :=\bigl\{ \gamma  \subset {{\mathbb
R}^d} \bigm| |\gamma  \cap \Lambda  |<\infty, \text{ for all }
\Lambda \in {\mathcal B}_{\mathrm{b}} ({{\mathbb R}^d})\bigr\}.
\end{equation}
The space $\Gamma $ is equipped with the vague topology, i.e., the
minimal topology for which all mappings $\Gamma \ni\gamma \mapsto
\sum_{x\in\gamma } f(x)\in{\mathbb R}$ are continuous for any
continuous function $f$ on ${{\mathbb R}^d}$ with compact support.
Note that $\sum_{x\in\gamma } f(x)$ is always finite in this case
since the summation is taken over only finitely many points of
$\gamma $ which belong to the support of $f$.
$\Gamma $ with the vague topology is a Polish space (see, e.g.,
\cite{KoKut} and references therein). The corresponding Borel
$\sigma $-algebra $ {\mathcal B}(\Gamma  )$ appears as the smallest
$\sigma $-algebra for which all mappings $\Gamma  \ni \gamma \mapsto
|\gamma_\Lambda |\in{\mathbb N}_0:={\mathbb N}\cup\{0\}$ are
measurable for any $\Lambda \in{\mathcal B}_{\mathrm{b}}({{\mathbb
R}^d})$. Here and below $\gamma_\Lambda :=\gamma \cap\Lambda $.

It can be shown that the space $\bigl( \Gamma , {\mathcal B}(\Gamma
)\bigr)$ is the projective limit of the family of spaces $\bigl\{(
\Gamma_\Lambda , {\mathcal B}(\Gamma_\Lambda ))\bigr\}_{\Lambda \in
{\mathcal B}_{\mathrm{b}} ({{\mathbb R}^d})}$. The Poisson measure
$\pi_\kappa$ on $\bigl(\Gamma  ,{\mathcal B}(\Gamma )\bigr)$ is
given as the projective limit of the family of measures
$\{\pi_\kappa ^\Lambda \}_{\Lambda \in {\mathcal B}_{\mathrm{b}}
({{\mathbb R}^d})}$, where $\pi_\kappa^\Lambda :=e^{-\kappa
m(\Lambda )}\lambda_\kappa $ is the probability measure on $\bigl(
\Gamma _\Lambda , {\mathcal B}(\Gamma_\Lambda )\bigr)$. Here
$m(\Lambda )$ is the Lebesgue measure of $\Lambda \in {\mathcal
B}_{\mathrm{b}} ({{\mathbb R}^d})$.

A function $F$ on $\Gamma $ is called {\it cylinder function} if it
may be characterized by the following relation:
\begin{equation}\label{cyl}
F(\gamma  )=F\upharpoonright_{\Gamma_\Lambda  }(\gamma_\Lambda)
\end{equation}
for some $\Lambda \in {\mathcal B}_{\mathrm{b}}({{\mathbb R}^d})$.
The class of all such function will be denoted by ${{\mathcal
F}_{\mathrm{cyl}}}(\Gamma )$.

A set $M\in {\mathcal B} (\Gamma_0)$ is called bounded if there
exists $\Lambda  \in {\mathcal B}_{\mathrm{b}} ({{\mathbb R}^d})$
and $N\in {\mathbb N}$ such that $M\subset \bigsqcup_{n=0}^N\Gamma
_\Lambda  ^{(n)}$. The set of bounded measurable functions with
bounded support we denote by ${B_{\mathrm{bs}}}(\Gamma_0)$, i.e.,
$G\in{B_{\mathrm{bs}}}(\Gamma_0)$ if $G\upharpoonright_{\Gamma
_0\setminus M}=0$ for some bounded $M\in {\mathcal B}(\Gamma_0)$.
Note that any ${\mathcal B}(\Gamma_0)$-measurable function $G$ on
$\Gamma_0$, in fact, is a sequence of functions
$\bigl\{G^{(n)}\bigr\}_{n\in{\mathbb N}_0}$, where $G^{(n)}$ is a
${\mathcal B}(\Gamma ^{(n)})$-measurable function on $\Gamma
^{(n)}$.

The following mapping between ${B_{\mathrm{bs}}} (\Gamma_0)$ and
${{\mathcal F}_{\mathrm{cyl}}}(\Gamma  )$ plays the key role in our
further considerations:
\begin{equation}
KG(\gamma  ):=\sum_{\eta \Subset \gamma  }G(\eta ), \quad \gamma
\in \Gamma , \label{KT3.15}
\end{equation}
where $G\in {B_{\mathrm{bs}}}(\Gamma_0)$, see, e.g.,
\cite{KoKu99,Le75I,Le75II}. The summation in the latter expression
is taken over all finite subconfigurations of $\gamma $, which is
denoted by the symbol $\eta \Subset \gamma  $. The mapping $K$ is
linear, positivity preserving, and invertible, with
\begin{equation}
K^{-1}F(\eta ):=\sum_{\xi \subset \eta }(-1)^{|\eta \setminus \xi
|}F(\xi ),\quad \eta \in \Gamma_0.\label{k-1trans}
\end{equation}
We denote the restriction of $K$ onto functions on $\Gamma_0$ by
$K_0$.

For any fixed $C>0$ we consider the following (pre-)norm on the
space ${B_{\mathrm{bs}}}(\Gamma_0)$
\begin{equation}\label{pre-norm}
\|G\|_C:= \int_{\Gamma_0} |G(\eta)| C^{|\eta|} \lambda (d\eta).
\end{equation}
The completion of ${B_{\mathrm{bs}}}(\Gamma_0)$ w.r.t. this pre-norm
is the following Banach space of ${\mathcal B}(\Gamma
_0)$-measurable functions
\begin{equation}\label{norm}
{\mathcal L}_C:=\bigl\{ G:\Gamma_0\rightarrow{\mathbb R} \bigm|
\|G\|_C <\infty\bigr\}.
\end{equation}

Let $\mu$ be a probability measure on $\bigl(\Gamma ,{\mathcal
B}(\Gamma )\bigr)$ such that $\int_\Gamma  |\gamma_\Lambda |^n \mu(d
\gamma )<\infty $ for any $\Lambda \in{\mathcal
B}_{\mathrm{b}}({{\mathbb R}^d})$, $n\in{\mathbb N}$. The class of
all such measures we denote by ${\mathcal M}^1_{\mathrm{fm}}(\Gamma
)$. A measure $\mu \in \mathcal{M}_{\mathrm{fm} }^1(\Gamma  )$ is
called locally absolutely continuous w.r.t. the Poisson measure
$\pi$ if for any $\Lambda  \in {\mathcal B}_{\mathrm{b}} ({{\mathbb
R}^d})$ the projection of $\mu$ onto $\Gamma_\Lambda $ is absolutely
continuous w.r.t. the projection of $\pi$ onto $\Gamma_\Lambda $. By
\cite{KoKu99}, in this case, there exists a system of measurable
symmetric functions $k_\mu^{(n)}: ({{\mathbb
R}^d})^n\rightarrow[0;+\infty)$, ${n\in{\mathbb N}}$ such that for
any $G\in{B_{\mathrm{bs}}}(\Gamma_0)$ the following identity holds
\begin{equation}\label{corfunc}
\int_\Gamma  (KG^{(n)})(\gamma ) \mu(d \gamma ) = \frac{1}{n!}
\int_{({{\mathbb R}^d})^n} G^{(n)}(x_1,\ldots, x_n)
k_\mu^{(n)}(x_1,\ldots, x_n) dx_1 \ldots d x_n.
\end{equation}
Functions $k_\mu^{(n)}$ are called the {\it correlation functions}
in mathematical physics as well as functions
$\frac{1}{n!}k_\mu^{(n)}$ are called the {\it factorial moments} in
probability theory.

We recall now without a proof the partial case of the well-known
technical lemma which plays a very important role in our
calculations (cf., \cite{KoMiZh}).
\begin{lemma}\label{Minlos}
For any measurable function $H:\Gamma_0\times\Gamma_0\times\Gamma
_0\rightarrow{\mathbb R}$
\begin{equation}\label{minlosid}
\int_{\Gamma_{0}}\sum_{\xi \subset \eta }H\left( \xi ,\eta \setminus
\xi ,\eta \right) \lambda \left( d \eta \right) =\int_{\Gamma
_{0}}\int_{\Gamma_{0}}H\left( \xi ,\eta ,\eta \cup \xi \right)
\lambda \left( d \xi \right) \lambda \left( d \eta \right)
\end{equation}
if both sides of the equality make sense.
\end{lemma}

\subsection{Glauber dynamics in continuum}

Let $\phi:{{\mathbb R}^d}\rightarrow{\mathbb R}_+:=[0;+\infty)$ be
even non-negative function which satisfies integrability condition
\begin{equation}\label{integrability}
C_\phi = \int_{{\mathbb R}^d} \bigl(1-e^{-\phi(x)}\bigr) dx
<+\infty.
\end{equation}
For any $\gamma \in\Gamma $, $x\in{{\mathbb R}^d}\setminus\gamma $
we set
\begin{equation}\label{relativeenergy}
E^\phi(x,\gamma ) :=\sum_{y\in\gamma } \phi(x-y) \in [0;\infty].
\end{equation}

Let us define the (pre-) generator of the Glauber dynamics: for any
$F\in{{\mathcal F}_{\mathrm{cyl}}}(\Gamma  )$ we set
\begin{align}
(LF)(\gamma ):=&\sum_{x\in\gamma } \bigl[F(\gamma \setminus x)
-F(\gamma )\bigr] \label{genGa}
\\&+ z \int_{{{\mathbb R}^d}} \bigl[F(\gamma \cup x)
-F(\gamma )\bigr]\exp\bigl\{-E^\phi(x,\gamma )\bigr\} dx, \qquad
\gamma \in\Gamma .\notag
\end{align}
Here $z>0$ is the {\it activity} parameter. Note that, because of
\eqref{cyl}, for $F\in{{\mathcal F}_{\mathrm{cyl}}}(\Gamma )$ there
exists $\Lambda \in{\mathcal B}_{\mathrm{b}}({{\mathbb R}^d})$ such
that $F(\gamma \setminus x)=F(\gamma )$ for any $x\in\gamma
_{\Lambda ^c}$ and $F(\gamma \cup x)=F(\gamma )$ for any
$x\in\Lambda ^c$; note also that $\exp\bigl\{-E^\phi(x,\gamma
)\bigr\}\leq 1$, therefore, the sum and integral in \eqref{genGa}
are finite.

Using the techniques considered in \cite{GarKur06}, it is possible
to show that there exists a~proper subspace
$\mathcal{S}\subset\Gamma $ and an~$\mathcal{S}$-valued stochastic
process with sample paths in the Skorokhod space
$D_\mathcal{S}[0;\infty)$ associated to the generator $L$.

This allows us to define the semigroup associated with $L$ in the
space of bounded continuous functions on $S$. This semigroup
determines the solution to the Kolmogorov equation, which formally
(only in the sense of action of operator) has the following form:
\begin{equation}
\frac{dF_t}{dt}=LF_t,\qquad F_t\bigm|_{t=0}=F_0.\label{Kolmogor}
\end{equation}

However, to show that $L$ is a generator of a semigroup in other
functional spaces on $\Gamma $ seems to be a difficult problem. This
difficulty is hidden in the complex structure of the non-linear
infinite dimensional space $\Gamma $.

In various applications the evolution of the corresponding
correlation functions (or measures) helps already to understand the
behavior of the process and gives candidates for invariant states.
The evolution of correlation functions of the process is related
heuristically to the evolution of states of our infinite particle
systems. The latter evolution is formally given as a solution to the
dual Kolmogorov equation (Fokker--Planck equation):
\begin{equation}
\frac{d\mu_t}{dt}=L^*\mu_t, \qquad
\mu_t\bigm|_{t=0}=\mu_0,\label{FokkerPlanc}
\end{equation}
where $L^*$ is the adjoint operator to $L$ on ${\mathcal
M}^1_\mathrm{fm}(\Gamma )$, provided, of course, that it exists.

Following the general scheme proposed in \cite{KoKutMi}, we
construct the evolution of functions which corresponds to the
\textit{symbol} ($K$-image) $\hat{L}=K^{-1}LK$ of the operator $L$
in $L^1$-space on $\Gamma_0$ w.r.t. the weighted Lebesgue--Poisson
measure, namely, in the space ${\mathcal L}_C$, see \eqref{norm}.

The evolution equation for \textit{quasi-observables} (functions on
$\Gamma_{0}$) corresponding to the Kolmogorov equation
(\ref{Kolmogor}) has the following form
\begin{equation}
\frac{dG_t}{dt}=\widehat{L}G_t,\qquad
G_t\bigm|_{t=0}=G_0.\label{quasiKolmogor}
\end{equation}
Then in a way analogous to that in which the corresponding
Fokker--Planck equation (\ref{FokkerPlanc}) was determined for
(\ref{Kolmogor}) we get the evolution equation for the correlation
functions corresponding to the equation (\ref{quasiKolmogor}):
\begin{equation}
\frac{dk_t}{dt}=\widehat{L}^{*}k_t,\qquad
k_t\bigm|_{t=0}=k_0,\label{corrfunctiona}
\end{equation}
where $\widehat{L}^{*}$ is the mapping dual to $\widehat{L}$ w.r.t.
the pairing
\begin{equation}\label{pairing}
\langle\!\langle G, k\rangle\!\rangle = \int_{\Gamma_0} G(\eta)
k(\eta) \lambda (d\eta).
\end{equation}

The existence of the evolution \eqref{quasiKolmogor} in ${\mathcal
L}_C$ gives now the following bounds for the solution of
\eqref{corrfunctiona} (if it exists):
\begin{equation}\label{RB}
|k_t(\eta)| \leq \mathrm{const} \cdot C^{|\eta|}, \quad
\eta\in\Gamma_0.
\end{equation}
The estimate \eqref{RB} is called the Ruelle bound: in \cite{R69},
\cite{R70}, it was shown that there is a class of Gibbs measures
$\{\mu\}$ whose correlation functions $\{k_\mu>0\}$ satisfy
\eqref{RB} for $\mathrm{const}=1$. The bound \eqref{RB} is also
called sub-Poissonian since $\{C^n\}_{n\geq0}$ is the system of the
correlation functions for the Poisson measure $\pi_C$.

In the present paper we obtain the strong solution to the equation
\eqref{quasiKolmogor} in ${\mathcal L}_C$. This allows us to solve
the equation \eqref{corrfunctiona} in a weak sense w.r.t. the
pairing \eqref{pairing}. In the forthcoming paper
\cite{FiKoKut-Glauber-2} we will consider the strong solution to
\eqref{corrfunctiona} in a proper Banach space. Moreover, we will
show that this solution at any moment of time is a correlation
function of some state.

\section{Construction and properties of the semigroup}

\subsection{Description of approximation}

Let $G\in {B_{\mathrm{bs}}}(\Gamma_0)$ then $F=KG\in{{\mathcal
F}_{\mathrm{cyl}}}(\Gamma )$. By \cite{FiKoOl07,KoKutZh}, we have
the following explicit form for the mapping $\hat{L}:=K^{-1}LK$ on
${B_{\mathrm{bs}}}(\Gamma_0)$
\begin{equation}\label{Lhat}
(\hat{L}G)(\eta) =- |\eta| G(\eta) + z
\sum_{\xi\subset\eta}\int_{{\mathbb R}^d} e^{-E^\phi(x,\xi)}
G(\xi\cup x)e_\lambda(e^{-\phi (x - \cdot)}-1,\eta\setminus\xi) dx ,
\end{equation}
where, by definition, for any $\mathcal{B}(\mathbb{R}^d)$-measurable
function $f$,
\begin{equation}\label{expLP}
e_\lambda (f,\eta ):=\prod_{x\in \eta }f(x) , \quad \eta \in \Gamma
_0\!\setminus\!\{\emptyset\}, \qquad  e_\lambda (f,\emptyset ):=1.
\end{equation}

Let us denote for any $\eta\in\Gamma_0$
\begin{align}
(L_0 G)(\eta) &:= - |\eta| G(\eta);\label{L0}\\
(L_1 G)(\eta) &:= z \sum_{\xi\subset\eta}\int_{{\mathbb R}^d}
e^{-E^\phi(x,\xi)} G(\xi\cup x)e_\lambda(e^{-\phi (x -
\cdot)}-1,\eta\setminus\xi) dx.\label{L1}
\end{align}

\begin{proposition}\label{prop-oper}
The expression \eqref{Lhat} defines a linear operator $\hat{L}$ in
${\mathcal L}_C$ with the dense domain ${\mathcal L}_{2C}\subset
{\mathcal L}_C$.
\end{proposition}
\begin{proof}
For any $G\in {\mathcal L}_{2C} $
\[
\left\Vert L_{0}G\right\Vert_{C} =\int_{\Gamma_0} |G(\eta)||\eta|
C^{|\eta|} \lambda (d \eta)<\int_{\Gamma_0} |G(\eta)|2^{|\eta|}
C^{|\eta|} \lambda (d\eta)<\infty
\]
and, by Lemma~\ref{Minlos},
\begin{align*}
\left\Vert L_{1}G\right\Vert_{C} &\leq z\int_{\Gamma
_{0}}\sum_{\xi \subset \eta }\int_{{{\mathbb
R}^d}}\,e^{-E^\phi(x,\xi )}\left\vert G(\xi \cup x)\right\vert
e_{\lambda }\left(\left\vert e^{-\phi (x-\cdot )}-1\right\vert ,\eta
\setminus \xi\right)dxC^{\left\vert \eta \right\vert }\lambda \left(
d\eta
\right)  \\
&=z\int_{\Gamma_{0}}\int_{\Gamma_{0}}\int_{{{\mathbb
R}^d}}\,e^{-E^\phi(x,\xi )}\left\vert G(\xi \cup x)\right\vert
e_{\lambda }\left(\left\vert e^{-\phi (x-\cdot )}-1\right\vert ,\eta
\right)dxC^{\left\vert \eta \right\vert }C^{\left\vert \xi
\right\vert }\lambda \left( d\xi \right) \lambda \left(d
\eta \right)  \\
&\leq \frac{z}{C}\exp \left\{ CC_{\phi }\right\} \int_{\Gamma
_{0}}\left\vert G\left( \xi \right) \right\vert \left\vert \xi
\right\vert C^{\left\vert \xi \right\vert }\lambda \left( d \xi
\right) < \frac{z}{C}\exp \left\{ CC_{\phi }\right\} \int_{\Gamma
_{0}}\left\vert G\left( \xi \right) \right\vert 2^{\left\vert \xi
\right\vert} C^{\left\vert \xi \right\vert }\lambda \left( d \xi
\right)  \\&<\infty .
\end{align*}
Embedding ${\mathcal L}_{2C}\subset{\mathcal L}_C$ is dense since
${B_{\mathrm{bs}}}(\Gamma_0)\subset{\mathcal L}_{2C}$.
\end{proof}

Let $\delta\in(0;1)$ be arbitrary and fixed. Consider for any
$\Lambda \in{\mathcal B}_{\mathrm{b}}({{\mathbb R}^d})$ the
following linear mapping on functions $F\in{{\mathcal
F}_{\mathrm{cyl}}}(\Gamma_0):=K_0 {B_{\mathrm{bs}}}(\Gamma_0)$
\begin{align}\label{apprGa0}
\left( P_{\delta }^\Lambda  F\right) \left( \gamma \right)
=&\sum_{\eta \subset \gamma }\delta ^{\left\vert \eta \right\vert
}\left( 1-\delta \right) ^{\left\vert \gamma \setminus \eta
\right\vert }\bigl(\Xi^\Lambda_\delta\left( \gamma
\right)\bigr)^{-1}
\\ &\times \int_{\Gamma_{\Lambda }}\left( z\delta \right)
^{\left\vert \omega \right\vert }\prod\limits_{y\in \omega
}e^{-E^\phi\left( y,\gamma \right) }F\left( \left( \gamma \setminus
\eta \right) \cup \omega \right) \lambda \left( d \omega
\right),\quad \gamma \in\Gamma_0,\notag
\end{align}
where
\begin{equation}\label{normfactor}
\Xi_\delta^\Lambda \left( \gamma \right) = \int_{\Gamma_{\Lambda
}}\left( z\delta \right) ^{\left\vert \omega \right\vert
}\prod\limits_{y\in \omega }e^{-E^\phi\left( y,\gamma \right)
}\lambda \left( d \omega \right).
\end{equation}
Clearly, $P_{\delta }^\Lambda $ is a positive preserving mapping and
\[
\left( P_{\delta }^\Lambda  1\right) \left( \gamma
\right)=\sum_{\eta \subset \gamma }\delta ^{\left\vert \eta
\right\vert }\left( 1-\delta \right) ^{\left\vert \gamma \setminus
\eta \right\vert }=1, \quad \gamma \in\Gamma_0.
\]

Operator \eqref{apprGa0} is constructed as a transition operator of
a Markov chain, which is a time discretization of a continuous time
process with the generator \eqref{genGa} and discretization
parameter $\delta\in(0;1)$. Roughly speaking, according to the
representation \eqref{apprGa0}, the probability of transition
$\gamma \rightarrow (\gamma \setminus\eta)\cup\omega$ (which
describes removing of subconfiguration $\eta\subset\gamma $ and
birth of a new subconfiguration $\omega\in\Gamma_\Lambda $) after
small time $\delta$ is equal to
\[
\bigl(\Xi_\delta^\Lambda (\gamma )\bigr)^{-1} \delta^{|\eta|}
(1-\delta)^{|\gamma \setminus\eta|}(z\delta)^{|\omega|}
\prod_{y\in\omega}e^{-E^\phi(y,\gamma )}.
\]

We may rewrite \eqref{apprGa0} in another manner.
\begin{proposition}
For any $F\in{{\mathcal F}_{\mathrm{cyl}}}(\Gamma_0)$ the following
equality holds
\begin{align}\label{anotherform}
\left( P_{\delta }^\Lambda  F\right) \left( \gamma \right) =&
\sum_{\xi \subset \gamma }\left( 1-\delta \right) ^{\left\vert \xi
\right\vert } \int_{\Gamma_{\Lambda }}\left( z\delta \right)
^{\left\vert \omega \right\vert }\prod\limits_{y\in \omega
}e^{-E^\phi\left( y,\gamma \right) }\\&\times (K_0^{-1}
F)\left( \xi \cup \omega \right) \lambda \left( d \omega
\right).\notag
\end{align}
\end{proposition}
\begin{proof} Let $G:=K_0^{-1}F\in {B_{\mathrm{bs}}}(\Gamma_0)$.
Since $\Xi_\delta^\Lambda $ doesn't depend on $\eta$, for $\gamma
\in\Gamma_0$ we have
\begin{align}\label{ex1}
\left( P_{\delta }^\Lambda  F\right) \left( \gamma \right)
=&\bigl(\Xi^\Lambda_\delta\left( \gamma \right)\bigr)^{-1}
\int_{\Gamma_{\Lambda }}\left( z\delta \right) ^{\left\vert \omega
\right\vert
}\prod\limits_{y\in \omega }e^{-E^\phi\left( y,\gamma \right) }\\
&\times \sum_{\eta \subset \gamma }\delta ^{\left\vert \gamma
\setminus \eta \right\vert } \left( 1-\delta \right) ^{\left\vert
\eta \right\vert } F\left( \eta \cup \omega \right) \lambda \left( d
\omega \right).\notag
\end{align}
To rewrite \eqref{apprGa0}, we have used also that any
$\eta\subset\gamma$ corresponds to a unique
$\gamma\setminus\eta\subset\gamma$. Applying the definition of $K_0$
to $F=K_0G$ we obtain
\begin{align}\label{ex2}
\sum_{\eta \subset \gamma }\delta ^{\left\vert \gamma \setminus
\eta \right\vert } \left( 1-\delta \right) ^{\left\vert \eta
\right\vert } F\left( \eta \cup \omega \right)&=\sum_{\eta \subset
\gamma }\delta ^{\left\vert \gamma \setminus \eta \right\vert }
\left(
1-\delta \right) ^{\left\vert\eta \right\vert } \sum_{\zeta\subset\eta}\sum_{\beta\subset\omega}G\left( \zeta \cup \beta \right)\\
&=\sum_{\zeta\subset\gamma}\sum_{\beta\subset\omega}G\left( \zeta \cup \beta \right)\sum_{\eta '\subset
\gamma\setminus\zeta }\delta ^{\left\vert \gamma \setminus( \eta'\cup\zeta) \right\vert }
\left(
1-\delta \right) ^{\left\vert\eta' \cup\zeta\right\vert },\notag
\end{align}
where after changing summation over $\eta\subset\gamma$ and
$\zeta\subset\eta$ we have used the fact that for any configuration
$\eta\subset\gamma$ which contains fixed $\zeta\subset\gamma$ there
exists a unique $\eta'\subset\gamma\setminus\zeta$ such that
$\eta=\eta'\cup\zeta$. But by the binomial formula
\begin{align}\label{ex3}
\sum_{\eta '\subset \gamma\setminus\zeta }\delta ^{\left\vert \gamma
\setminus( \eta'\cup\zeta) \right\vert } \left( 1-\delta \right)
^{\left\vert\eta' \cup\zeta\right\vert
}&=(1-\delta)^{|\zeta|}\sum_{\eta '\subset \gamma\setminus\zeta
}\delta ^{\left\vert (\gamma\setminus \zeta) \setminus \eta'
\right\vert } \left( 1-\delta \right) ^{\left\vert\eta'\right\vert }
\\&=(1-\delta)^{|\zeta|} (\delta+1-\delta)^{|\gamma
\setminus\zeta|}=(1-\delta)^{|\zeta|}.\notag
\end{align}
Combining \eqref{ex1}, \eqref{ex2}, \eqref{ex3}, we get
\begin{align*}
\left( P_{\delta }^\Lambda  F\right) \left( \gamma \right)\notag
=&\bigl(\Xi^\Lambda_\delta\left( \gamma \right)\bigr)^{-1}
\int_{\Gamma_{\Lambda }}\left( z\delta \right) ^{\left\vert \omega
\right\vert
}\prod\limits_{y\in \omega }e^{-E^\phi\left( y,\gamma \right) }\\
&\times
\sum_{\zeta\subset\gamma}\sum_{\beta\subset\omega}G\left( \zeta \cup
\beta \right) (1-\delta)^{|\zeta|} \lambda \left( d \omega \right).
\end{align*}
Next, Lemma~\ref{Minlos} yields
\begin{align*}
\left( P_{\delta }^\Lambda  F\right) \left( \gamma
\right)=&\bigl(\Xi^\Lambda_\delta\left( \gamma \right)\bigr)^{-1}
\int_{\Gamma_{\Lambda }}\int_{\Gamma_{\Lambda }}\left( z\delta
\right) ^{\left\vert \omega \cup\beta\right\vert
}\prod\limits_{y\in \omega\cup\beta }e^{-E^\phi\left( y,\gamma \right) }\notag\\
&\times \sum_{\zeta\subset\gamma}G\left( \zeta \cup \beta
\right) (1-\delta)^{|\zeta|}\lambda \left( d \omega \right) \lambda
\left( d \beta \right)\notag\\=&\int_{\Gamma_{\Lambda }}\left(
z\delta \right) ^{\left\vert \beta\right\vert }\prod\limits_{y\in
\beta }e^{-E^\phi\left( y,\gamma \right) }
\sum_{\zeta\subset\gamma}G\left( \zeta \cup \beta \right)
(1-\delta)^{|\zeta|} \lambda \left( d \beta \right),\notag
\end{align*}
which proves the statement.
\end{proof}

In the next proposition we describe the image of $P_\delta^\Lambda $
under the $K_0$-transform.
\begin{proposition}
Let $\hat{P}_\delta^\Lambda =K_0^{-1}P_\delta^\Lambda  K_0$. Then
for any $G\in {B_{\mathrm{bs}}}(\Gamma_0)$ the following equality
holds
\begin{align}\label{apprsemigroupLa}
\bigl(\hat{P}_{\delta }^\Lambda  G\bigr) \left( \eta \right)
=&\sum_{\xi \subset \eta }\left( 1-\delta \right) ^{\left\vert \xi
\right\vert }\int_{\Gamma_{\Lambda }}\left( z\delta \right)
^{\left\vert \omega \right\vert }G\left( \xi \cup \omega \right)
\\&\times \prod\limits_{y\in \xi }e^{-E^\phi\left( y,\omega \right)
}\prod\limits_{y'\in \eta \setminus \xi }\left( e^{-E^\phi\left(
y',\omega \right) }-1\right) \lambda \left( d \omega \right) , \quad
\eta\in\Gamma_0.\notag
\end{align}
\end{proposition}
\begin{proof}
By \eqref{anotherform} and the definition of $K_0^{-1}$, we have
\begin{align*}
\bigl(\hat{P}_{\delta }^\Lambda  G\bigr) \left( \eta \right)
&=\sum_{\zeta\subset\eta}(-1)^{|\eta\setminus\zeta|}\sum_{\xi
\subset\zeta }\left( 1-\delta \right) ^{\left\vert \xi \right\vert }
\int_{\Gamma_{\Lambda }}\left( z\delta \right) ^{\left\vert \omega
\right\vert }\prod\limits_{y\in \omega }e^{-E^\phi\left( y,\zeta
\right) }G\left( \xi \cup \omega \right) \lambda \left( d \omega
\right)\\&=\sum_{\xi\subset\eta} \left( 1-\delta \right)
^{\left\vert \xi \right\vert }\sum_{\zeta\subset\eta\setminus\xi }
(-1)^{|(\eta\setminus\xi)\setminus\zeta|}\int_{\Gamma_{\Lambda
}}\left( z\delta \right) ^{\left\vert \omega \right\vert
}\prod\limits_{y\in \omega }e^{-E^\phi\left( y,\zeta\cup\xi \right)
}G\left( \xi \cup \omega \right) \lambda \left( d \omega \right).
\end{align*}
Using the definition \eqref{relativeenergy} of the relative energy
we obtain
\[
\prod\limits_{y\in \omega }e^{-E^\phi\left( y,\zeta\cup\xi
\right)}=\prod\limits_{y\in \xi }e^{-E^\phi\left( y,\omega\right)}
\prod\limits_{y'\in \zeta }e^{-E^\phi\left( y',\omega \right)}.
\]
The well-known equality
\begin{align*}
\sum_{\zeta\subset\eta\setminus\xi } (-1)^{|(\eta
\setminus\xi)\setminus\zeta|}\prod\limits_{y'\in \zeta
}e^{-E^\phi\left( y',\omega
\right)}&=\Bigl(K_0^{-1}\prod\limits_{y'\in\cdot }e^{-E^\phi\left(
y',\omega \right)}\Bigl)(\eta\setminus\xi)\\&=\prod\limits_{y'\in
\eta \setminus \xi }\left( e^{-E^\phi\left( y',\omega \right)
}-1\right)
\end{align*}
(see, e.g., \cite{FiKoOl07}) completes the proof.
\end{proof}

\subsection{Construction of the semigroup on ${\mathcal L}_C$}

By analogy with \eqref{apprsemigroupLa}, we consider the following
linear mapping on measurable functions on $\Gamma_0$
\begin{align}\label{apprsemigroup}
\bigl(\hat{P}_{\delta } G\bigr) \left( \eta \right) := & \sum_{\xi
\subset \eta }\left( 1-\delta \right) ^{\left\vert \xi \right\vert
}\int_{\Gamma_{0}}\left( z\delta \right) ^{\left\vert \omega
\right\vert }G\left( \xi \cup \omega \right)
\\&\times \prod\limits_{y\in \xi }e^{-E^\phi\left( y,\omega \right)
}\prod\limits_{y'\in \eta \setminus \xi }\left( e^{-E^\phi\left(
y',\omega \right) }-1\right) \lambda \left( d \omega \right) , \quad
\eta\in\Gamma_0.\notag
\end{align}

\begin{proposition}\label{prop_contr}
Let
\begin{equation}
ze^{CC_{\phi }} \leq C. \label{smallparam}
\end{equation}
Then $\hat{P}_\delta$, given by \eqref{apprsemigroup}, is a well
defined linear operator in ${\mathcal L}_C$, such that
\begin{equation}
\bigl\| \hat{P}_\delta \bigr\|
\leq 1. \label{contacrion}
\end{equation}
\end{proposition}
\begin{proof}
Since $\phi \geq 0$ we have
\begin{align*}
\left\Vert \hat{P}_{\delta }G\right\Vert_{C}  \leq &\int_{\Gamma
_{0}}\sum_{\xi \subset \eta }\left( 1-\delta \right) ^{\left\vert
\xi \right\vert }\int_{\Gamma_{0}}\left( z\delta \right)
^{\left\vert \omega \right\vert }\left\vert G\left( \xi \cup \omega
\right) \right\vert \\ & \times \prod\limits_{y\in \xi
}e^{-E^\phi\left( y,\omega \right) }\prod\limits_{y'\in \eta
\setminus \xi }\left\vert e^{-E^\phi\left( y',\omega \right)
}-1\right\vert \lambda \left( d \omega \right) C ^{\left\vert
\eta \right\vert } \lambda \left( d \eta \right) \\
= &\int_{\Gamma_{0}}\int_{\Gamma_{0}}\left( 1-\delta \right)
^{\left\vert \xi \right\vert }\int_{\Gamma_{0}}\left( z\delta
\right) ^{\left\vert \omega \right\vert }\left\vert G\left( \xi \cup
\omega \right) \right\vert \\ &\times \prod\limits_{y\in \xi
}e^{-E^\phi\left( y,\omega \right) }\prod\limits_{y'\in \eta
}\left\vert e^{-E^\phi\left( y',\omega \right) }-1\right\vert
\lambda \left( d \omega \right) C ^{\left\vert \eta \right\vert }C
^{\left\vert \xi \right\vert } \lambda \left( d \xi
\right) \lambda \left( d \eta \right) \\
=&\int_{\Gamma_{0}}\int_{\Gamma_{0}}\left( 1-\delta \right)
^{\left\vert \xi \right\vert }\left( z\delta \right) ^{\left\vert
\omega \right\vert }\left\vert G\left( \xi \cup \omega \right)
\right\vert \\ &\times \prod\limits_{y\in \xi }e^{-E^\phi\left(
y,\omega \right) }\exp \left\{ C\int_{{{\mathbb R}^d}}\left(
1-e^{-E^\phi\left( y',\omega \right) }\right) dy'\right\} \lambda
\left( d \omega \right) C ^{\left\vert \xi \right\vert } \lambda
\left( d \xi \right).
\end{align*}

It is easy to see by the induction principle that for $\phi \geq 0$,
$\omega\in\Gamma_0$, $y\notin\omega$
\begin{equation}\label{keypos}
1-e^{-E^\phi\left( y,\omega \right) }=1-\prod\limits_{x\in \omega
}e^{-\phi \left( x-y\right) }\leq \sum_{x\in \omega }\left(
1-e^{-\phi \left( x-y\right) }\right).
\end{equation}
Then
\begin{align*}
\bigl\Vert \hat{P}_{\delta }G\bigr\Vert_{C} \leq &\int_{\Gamma
_{0}}\int_{\Gamma_{0}}\left( 1-\delta \right) ^{\left\vert \xi
\right\vert }\left( z\delta \right) ^{\left\vert \omega \right\vert
}\left\vert G\left( \xi \cup \omega \right) \right\vert \\
&\times \exp \left\{ C\sum_{x\in \omega }\int_{{{\mathbb
R}^d}}\left( 1-e^{-\phi \left( x-y\right) }\right) dy\right\}
\lambda \left( d \omega \right) C ^{\left\vert \xi
\right\vert } \lambda \left( d \xi \right) \\
=&\int_{\Gamma_{0}}\int_{\Gamma_{0}}\left( 1-\delta \right)
^{\left\vert \xi \right\vert }\left( z\delta \right) ^{\left\vert
\omega \right\vert }\left\vert G\left( \xi \cup \omega \right)
\right\vert e^{CC_{\phi }\left\vert \omega \right\vert }C
^{\left\vert \xi \right\vert
} \lambda \left( d \omega \right) \lambda \left( d \xi \right) \\
=&\int_{\Gamma_{0}}\left[ \left( 1-\delta \right) C+z\delta
e^{CC_{\phi }} \right] ^{\left\vert \omega \right\vert }\left\vert
G\left( \omega \right) \right\vert \lambda \left( d \omega \right)
\leq \left\Vert G\right\Vert_{C}.
\end{align*}
For the last inequality we have used that \eqref{smallparam} implies
$\left( 1-\delta \right) C+z\delta e^{CC_{\phi }}\leq C$. Note that,
for $\lambda $-a.a. $\eta\in\Gamma_0$
\begin{equation}\label{correctness}
\bigl(\hat{P}_\delta G \bigr) (\eta) <\infty,
\end{equation}
and the statement is proved.
\end{proof}

\begin{proposition}\label{prop_ineq}
Let the inequality \eqref{smallparam} be fulfilled and define
\[
\hat{L}_\delta:=\frac{1}{\delta}(\hat{P}_\delta-1\!\!1), \quad
\delta\in(0;1),
\]
where $1\!\!1$ is the identity operator in ${\mathcal L}_C$. Then
for any $G\in {\mathcal L}_{2C}$
\begin{equation}\label{ineq}
\bigl\| (\hat{L}_\delta -\hat{L})G\bigr\|_C\leq 3\delta \|G\|_{2C}.
\end{equation}
\end{proposition}
\begin{proof}

Let us denote
\begin{align}
\bigl( \hat{P}_{\delta }^{\left( 0\right) }G\bigr) \left( \eta
\right) =&\sum_{\xi \subset \eta }\left( 1-\delta \right)
^{\left\vert \xi \right\vert }G\left( \xi \right) 0^{\left\vert \eta
\setminus \xi \right\vert }=\left( 1-\delta \right) ^{\left\vert
\eta \right\vert }G\left(
\eta \right) ; \label{P0}\\
\bigl( \hat{P}_{\delta }^{\left( 1\right) }G\bigr) \left( \eta
\right) =& \, z\delta \sum_{\xi \subset \eta }\left( 1-\delta \right)
^{\left\vert \xi \right\vert }\int_{{{\mathbb R}^d}}G\left( \xi \cup
x\right)\\&\times \prod\limits_{y\in \xi }e^{-\phi \left(
y-x\right) }\prod\limits_{y\in \eta \setminus \xi }\left( e^{-\phi
\left( y-x\right) }-1\right) dx; \label{P1}
\end{align}
and
\begin{equation}
\hat{P}_{\delta }^{\left( \geq 2\right) }=\hat{P}_{\delta }-\left(
\hat{P}_{\delta }^{\left( 0\right) }+\hat{P}_{\delta }^{\left(
1\right) }\right) \label{P2}.
\end{equation}

Clearly
\begin{align}\label{trin}
\bigl\| (\hat{L}_\delta -\hat{L})G\bigr\|_C&=\left\Vert
\frac{1}{\delta }\left( \hat{P}_{\delta }G-G\right) -\hat{L} G\right\Vert_{C} \\
&\leq \left\Vert \frac{1}{\delta }\left( \hat{P}_{\delta }^{\left(
0\right) }G-G\right) -L_{0}G\right\Vert_{C}+\left\Vert
\frac{1}{\delta }\hat{P}_{\delta }^{\left( 1\right)
}G-L_{1}G\right\Vert_{C}+\frac{1}{\delta } \left\Vert
\hat{P}_{\delta }^{\left( \geq 2\right) }G\right\Vert_{C}.\notag
\end{align}
Now we estimate each of the terms in \eqref{trin} separately. By
\eqref{L0} and \eqref{P0}, we have
\begin{equation*}
\left\Vert \frac{1}{\delta }\left( \hat{P}_{\delta }^{\left(
0\right) }G-G\right) -L_{0}G\right\Vert_{C}  =\int_{\Gamma
_{0}}\left\vert \frac{\left( 1-\delta \right) ^{\left\vert \eta
\right\vert }-1}{\delta }+\left\vert \eta \right\vert \right\vert
\left\vert G\left( \eta \right) \right\vert C^{\left\vert \eta
\right\vert }\lambda \left( d \eta \right).
\end{equation*}
But, for any $|\eta|\geq2$
\begin{align*}
\left\vert \frac{\left( 1-\delta \right) ^{\left\vert \eta
\right\vert }-1}{\delta }
+\left\vert \eta \right\vert \right\vert &=\left\vert\sum_{k=2}^{|\eta|} \binom{|\eta|}{k}(-1)^{k}\delta^{k-1}\right\vert\\
&=\delta\left\vert\sum_{k=2}^{|\eta|} \binom{|\eta|}{k}(-1)^{k}\delta^{k-2}\right\vert\leq \delta \sum_{k=2}^{|\eta|} \binom{|\eta|}{k} <\delta\cdot 2^{|\eta|}.
\end{align*}
Therefore,
\begin{equation}\label{est0}
\left\Vert \frac{1}{\delta }\left( \hat{P}_{\delta }^{\left(
0\right) }G-G\right) -L_{0}G\right\Vert_{C}\leq \delta \|G\|_{2C}.
\end{equation}

Next, by \eqref{L1} and \eqref{P1}, one can write
\begin{align*}
\left\Vert \frac{1}{\delta }\hat{P}_{\delta }^{\left( 1\right)
}G-L_{1}G\right\Vert_{C}  = & \, z\int_{\Gamma_{0}}\biggl\vert \sum_{\xi \subset \eta }\left(
\left( 1-\delta \right) ^{\left\vert \xi \right\vert }-1\right)
\int_{{{\mathbb{R}} ^{d}}}G\left( \xi \cup x\right)
\prod\limits_{y\in \xi }e^{-\phi \left( y-x\right) }\\& \times
\prod\limits_{y\in \eta \setminus \xi }\left( e^{-\phi \left(
y-x\right) }-1\right) dx\biggr\vert C^{\left\vert \eta \right\vert
}\lambda
\left( d \eta \right)  \\
\leq & \, z\int_{\Gamma_{0}}\int_{\Gamma_{0}}\left( 1-\left( 1-\delta
\right) ^{\left\vert \xi \right\vert }\right)
\int_{{{\mathbb{R}}^{d}}}\left\vert G\left( \xi \cup x\right)
\right\vert \prod\limits_{y\in \xi }e^{-\phi \left( y-x\right)
}\\& \times \prod\limits_{y\in \eta }\left( 1-e^{-\phi \left(
y-x\right) }\right) dxC^{\left\vert \xi \right\vert }C^{\left\vert
\eta \right\vert } \lambda \left( d \xi \right) \lambda \left( d
\eta \right) ,
\end{align*}
where we have used Lemma~\ref{Minlos}. Note that for any $\left\vert
\xi \right\vert \geq 1$
\begin{equation*}
1-\left( 1-\delta \right) ^{\left\vert \xi \right\vert }=\delta
\sum_{k=0}^{\left\vert \xi \right\vert -1}\left( 1-\delta \right)
^{k}\leq \delta \left\vert \xi \right\vert
\end{equation*}
Then, by \eqref{smallparam} and \eqref{integrability}, one may
estimate
\begin{align}
\left\Vert \frac{1}{\delta }\hat{P}_{\delta }^{\left( 1\right)
}G-L_{1}G\right\Vert_{C} \label{spec1} &\leq z\delta \int_{\Gamma_{0}}\left\vert \xi \right\vert
\int_{{{\mathbb{R}}^{d}}}\left\vert G\left( \xi \cup x\right)
\right\vert dxC^{\left\vert \xi \right\vert }e^{CC_{\phi }}\lambda
\left( d \xi \right)   \\
&\leq z \delta \int_{\Gamma_{0}}\left\vert \xi \right\vert \left(
\left\vert \xi \right\vert -1\right) \left\vert G\left( \xi \right)
\right\vert C^{\left\vert \xi \right\vert -1}e^{CC_{\phi }}\lambda
\left( d \xi \right) . \notag
\end{align}
Since $n\left( n-1\right) \leq 2^{n}$, $n\geq 1$ and by
\eqref{smallparam}, the latter expression can be bounded by
\[
\delta \int_{\Gamma_{0}}\left\vert G\left( \xi \right) \right\vert
\left( 2C\right) ^{\left\vert \xi \right\vert }\lambda \left(  d\xi
\right).
\]

Finally, Lemma~\ref{Minlos}, \eqref{keypos} and bound
$e^{-E^\phi(y,\omega)}\leq 1$, imply (set $\Gamma_{0}^{\left( \geq
2\right) }:=\bigsqcup_{n\geq2}\Gamma ^{(n)}$)
\begin{align}
\left\Vert \frac{1}{\delta }\hat{P}_{\delta }^{\left( \geq 2\right)
}G\right\Vert_{C} \leq & \, \frac{1}{\delta }\int_{\Gamma_{0}}\sum_{\xi
\subset \eta }\left( 1-\delta \right) ^{\left\vert \xi \right\vert
}\int_{\Gamma_{0}^{\left( \geq 2\right) }}\left( z\delta \right)
^{\left\vert \omega \right\vert }\left\vert G\left(
\xi \cup \omega \right) \right\vert \label{spec01}\\
& \times \prod\limits_{y\in \xi }e^{-E^\phi\left( y,\omega
\right) }\prod\limits_{y\in \eta \setminus \xi }\left(
1-e^{-E^\phi\left( y,\omega \right) }\right) \lambda \left( d \omega
\right) C ^{\left\vert \eta \right\vert }\lambda \left( d \eta
\right) \notag
\\
\leq & \, \delta \int_{\Gamma_{0}}\sum_{\xi \subset \eta }\left(
1-\delta \right) ^{\left\vert \xi \right\vert }\int_{\Gamma
_{0}^{\left( \geq 2\right) }}z^{\left\vert \omega \right\vert
}\left\vert G\left( \xi \cup \omega \right) \right\vert \notag
\\
& \times \prod\limits_{y\in \xi }e^{-E^\phi\left( y,\omega
\right) }\prod\limits_{y\in \eta \setminus \xi }\left(
1-e^{-E^\phi\left( y,\omega \right) }\right) \lambda \left( d \omega
\right) C
^{\left\vert \eta \right\vert } \lambda \left( d \eta \right)  \notag\\
\leq &  \, \delta \int_{\Gamma_{0}}\sum_{\xi \subset \eta }\left(
1-\delta \right) ^{\left\vert \xi \right\vert }\int_{\Gamma
_{0}}z^{\left\vert \omega \right\vert }\left\vert G\left( \xi \cup
\omega \right) \right\vert \notag\\
&\times \prod\limits_{y\in \xi }e^{-E^\phi\left( y,\omega
\right) }\prod\limits_{y\in \eta \setminus \xi }\left(
1-e^{-E^\phi\left( y,\omega \right) }\right) \lambda \left( d \omega
\right) C ^{\left\vert \eta \right\vert }\lambda \left(
d \eta \right)  \notag\\
\leq & \, \delta \int_{\Gamma_{0}}\int_{\Gamma_{0}}\left( 1-\delta
\right) ^{\left\vert \xi \right\vert }z^{\left\vert \omega
\right\vert }\left\vert G\left( \xi \cup
\omega \right) \right\vert \notag\\
& \times \int_{\Gamma_{0}}\prod\limits_{y\in \eta }\left(
1-e^{-E^\phi\left( y,\omega \right) }\right)  C ^{\left\vert \eta
\right\vert }\lambda \left( d \eta \right) \lambda \left( d \omega
\right) C ^{\left\vert \xi \right\vert }d\lambda \left( \xi \right)
\notag\\
\leq &  \, \delta \int_{\Gamma_{0}}\int_{\Gamma_{0}}\left( 1-\delta
\right) ^{\left\vert \xi \right\vert }z^{\left\vert \omega
\right\vert }\left\vert G\left( \xi \cup \omega \right) \right\vert
e^{CC_\phi\vert\omega\vert} d\lambda \left( \omega \right) C
^{\left\vert \xi \right\vert }d\lambda \left( \xi \right) \notag\\
\leq & \, \delta \int_{\Gamma_{0}}\left[ \left( 1-\delta \right)
C+ze^{CC_{\phi }}\right] ^{\left\vert \omega \right\vert }\left\vert
G\left(
\omega \right) \right\vert d\lambda \left( \omega \right)  \notag\\
\leq & \,  \delta \int_{\Gamma_{0}}\left[ \left( 2-\delta \right)
C\right] ^{\left\vert \omega \right\vert }\left\vert G\left( \omega
\right) \right\vert d\lambda \left( \omega \right) \leq \delta
\int_{\Gamma_{0}}\left\vert G\left( \omega \right) \right\vert
\left( 2C\right) ^{\left\vert \omega \right\vert }d\lambda \left(
\omega \right) .\notag
\end{align}
Combining inequalities \eqref{est0}--\eqref{spec01} we obtain the
assertion of the proposition.
\end{proof}

We will need the following results in the sequel.

\begin{lemma}[{\cite[Corollary 3.8]{EK}}] \label{EK_res}
Let $A$ be a linear operator on a~Banach space $L$ with $D\left(
A\right) $ dense in $L$, and let $|\!|\!|\cdot |\!|\!|$ be a norm on
$D\left( A\right) $ with respect to which $D\left( A\right) $ is a
Banach space. For $n\in \mathbb{N}$ let $T_{n}$ be a linear
$\left\Vert \cdot \right\Vert $-contraction on $L$ such that
$T_{n}:D\left( A\right) \rightarrow D\left( A\right) $, and define
$A_{n}=n\left( T_{n}-1\right) $. Suppose there exist $\omega \geq 0$
and a sequence $\left\{ \varepsilon_{n}\right\} \subset \left(
0;+\infty \right) $ tending to zero such that for $n\in \mathbb{N}$
\begin{equation}\label{approperEK}
\left\Vert \left( A_{n}-A\right) f\right\Vert \leq \varepsilon_{n}
|\!|\!| f |\!|\!|,~f\in D\left( A\right)
\end{equation}
and
\begin{equation}\label{psevdocontr}
\bigl|\!\bigl|\!\bigl| T_{n}\upharpoonright_{D(A)}
\bigr|\!\bigr|\!\bigr| \leq 1+\frac{\omega }{n}.
\end{equation}
Then $A$ is closable and the closure of $A$ generates a strongly
continuous contraction semigroup on $L$.
\end{lemma}

\begin{lemma}[cf. {\cite[Theorem 6.5]{EK}}] \label{EK_res-conv}
Let $L, L_n$, $n\in{\mathbb N}$ be Banach spaces, and $p_n:
L\rightarrow L_n$ be bounded linear transformation, such that
$\sup_n \|p_n\|<\infty $. For any $n\in{\mathbb N}$, let $T_n$ be a
linear contraction on $L_n$, let $\varepsilon_n>0$ be such that
$\lim_{n\rightarrow \infty} \varepsilon_n =0$, and put
$A_n=\varepsilon_n^{-1}(T_n - 1\!\!1)$. Let $T_t$ be a strongly
continuous contraction semigroup on $L$ with generator $A$ and let
$D$ be a core for $A$. Then the following are equivalent:
\begin{enumerate}
\item For each $f\in L$, $T_n^{[t/\varepsilon_n]} p_n f\rightarrow p_n
T_t f$ in $L_n$ for all $t\geq0$ uniformly on bounded intervals.
Here and below $[\,\cdot\,\,]$ mean the entire part of a real
number.

\item For each $f\in D$, there exists $f_n\in L_n$ for each
$n\in{\mathbb N}$ such that $f_n \rightarrow p_n f$ and $A_n f_n
\rightarrow p_n Af$ in $L_n$.
\end{enumerate}
\end{lemma}

And now we are able to show the existence of the semigroup on
${\mathcal L}_C$.

\begin{theorem}\label{semigroup0}
Let
\begin{equation}\label{verysmallparam}
    z\leq \min\bigl\{Ce^{-CC_{\phi }} ; 2Ce^{-2CC_{\phi }}\bigr\} .
\end{equation}
Then $\bigl(\hat{L}, {\mathcal L}_{2C}\bigr)$ from
Proposition~\ref{prop-oper} is a closable linear operator in
${\mathcal L}_C$ and its closure $\bigl(\hat{L}, D(\hat{L})\bigr)$
generates a strongly continuous contraction semigroup $\hat{T}_t$ on
${\mathcal L}_C$.
\end{theorem}
\begin{proof} We apply Lemma~\ref{EK_res} for $L={\mathcal L}_C$,
$\bigl(A,D(A)\bigr)=\bigl(\hat{L}, {\mathcal L}_{2C}\bigr)$,
$|\!|\!| \cdot |\!|\!| :=\|\cdot\|_{2C}$; $T_{n}=\hat{P}_{\delta}$
and $A_{n}=n\left( T_{n}-1\right) =\frac{1}{\delta } (
\hat{P}_{\delta }-1\!\!1)=\hat{L}_\delta$, where
$\delta=\frac{1}{n}$, $n\geq 2$.

Condition $ze^{CC_{\phi }}\leq C$, Proposition~\ref{prop_contr}, and
Proposition~\ref{prop_ineq} provide that $T_n$, $n\geq 2$ are linear
$\|\cdot\|_C$-contractions and \eqref{approperEK} holds with
$\varepsilon_n=\frac{3}{n}=3\delta$. On the other hand, in addition,
Proposition~\ref{prop_contr} applied to the constant $2C$ instead of
$C$ gives \eqref{psevdocontr} for $\omega=0$ under condition
$ze^{2CC_{\phi }}\leq 2C$.
\end{proof}

Moreover, since we proved the existence of the semigroup $\hat{T}_t$
on ${\mathcal L}_C$ one can apply contractions $\hat{P}_\delta$
defined above by \eqref{apprsemigroup} to approximate the semigroup
$\hat{T}_t$.
\begin{corollary}\label{approx0}
Let \eqref{smallparam} holds. Then for any $G\in{\mathcal L}_C$
\[
\bigl(\hat{P}_{\frac{1}{n}}\bigr)^{[nt]} G \rightarrow \hat{T}_t G,
\quad n\rightarrow \infty
\]
for all $t\geq0$ uniformly on bounded intervals.
\end{corollary}
\begin{proof}
The statement is a direct consequence of Theorem~\ref{semigroup0},
convergence \eqref{ineq}, and Lemma~\ref{EK_res-conv} (if we set
$L_n=L={\mathcal L}_C$, $p_n=1\!\!1$, $n\in{\mathbb N}$).
\end{proof}

\subsection{Finite-volume approximation of $\hat{T}_t$}

Note that $\hat{P}_\delta$ defined by \eqref{apprsemigroup} is a
formal point-wise limit of $\hat{P}_\delta^\Lambda $ as $\Lambda
\uparrow {{\mathbb R}^d}$. We have shown in \eqref{correctness} that
this definition is correct. Corollary~\ref{approx0} claims
additionally that the linear contractions $\hat{P}_\delta$
approximate the semigroup $\hat{T}_t$, when $\delta\downarrow 0$.
One may also show that mappings $\hat{P}_\delta^\Lambda $ have
a~similar property when $\Lambda \uparrow{{\mathbb R}^d}$,
$\delta\downarrow 0$.

Let us fix a system $\{\Lambda_n\}_{n\geq 2}$, where $\Lambda
_n\in{\mathcal B}_{\mathrm{b}}({{\mathbb R}^d})$, $\Lambda
_n\subset\Lambda_{n+1}$, $\bigcup_{n} \Lambda_n={{\mathbb R}^d}$. We
set
\[
T_n:=\hat{P}_{\frac{1}{n}}^{\Lambda_n}.
\]
Note that any $T_n$ is a linear mapping on ${B_{\mathrm{bs}}}(\Gamma
_0)$. We consider also the system of Banach spaces of measurable
functions on $\Gamma_0$
\[
{\mathcal L}_{C,n}:=\biggl\{ G:\Gamma_{\Lambda_n}\rightarrow{\mathbb
R} \biggm| \|G\|_{C,n}:= \int_{\Gamma_{\Lambda_n}} |G(\eta)|
C^{|\eta|} \lambda (d \eta) <\infty\biggr\}.
\]
Let $p_n:{\mathcal L}_C\rightarrow{\mathcal L}_{C,n}$ be a cut-off
mapping, namely, for any $G\in{\mathcal L}_C$
\[
(p_n G)(\eta) = 1\!\!1_{\Gamma_{\Lambda_n}} (\eta) G(\eta).
\]
Then, obviously, $\|p_n G\|_{C,n}\leq \|G\|_C$. Hence,
$p_n:{\mathcal L}_C\rightarrow{\mathcal L}_{C,n}$ is a linear
bounded transformation with $\|p_n\|=1$.

\begin{proposition}\label{approxn}
Let \eqref{smallparam} hold. Then for any $G\in{\mathcal L}_C$
\[
\bigl\| \bigl(T_n\bigr)^{[nt]} p_n G - p_n \hat{T}_t
G\bigr\|_{C,n}\rightarrow 0, \quad n\rightarrow \infty
\]
for all $t\geq0$ uniformly on bounded intervals.
\end{proposition}
\begin{proof}
The proof of the proposition is completed by showing that all
conditions of Lemma~\ref{EK_res-conv} hold. Using completely the
same arguments as in the proof of Proposition~\ref{prop_contr} one
gets that each $T_n=\hat{P}_{\frac{1}{n}}^{\Lambda_n}$ is a linear
contraction on ${\mathcal L}_{C,n}$, $n\geq 2$ (note that for any
$n\geq 2$, \eqref{integrability} implies $\int_{\Lambda
_n}\bigl(1-e^{-\phi(x)}\bigr)dx\leq C_\phi<\infty$). Next, we set
$A_n=n (T_n - 1\!\!1_n)$ where $1\!\!1_n$ is a unit operator on
${\mathcal L}_{C,n}$ and let us expand $T_n$ in three parts
analogously to the proof of Proposition~\ref{prop_ineq}:
$T_n=T_n^{(0)}+T_n^{(1)}+T_n^{(\geq2)}$. As a result, $A_n=n
(T_n^{(0)} - 1\!\!1_n) + nT_n^{(1)}+nT_n^{(\geq2)}$. For any $G\in
{\mathcal L}_{2C}$ we set $G_n=p_n G\in{\mathcal
L}_{2C,n}\subset{\mathcal L}_{C,n}$. To finish the proof we have to
verify that for any $G\in {\mathcal L}_{2C}$
\begin{equation}\label{safcond}
\| A_n G_n - p_n \hat{L}G \|_{C,n} \rightarrow 0, \quad n\rightarrow
\infty.
\end{equation}
For any $G\in {\mathcal L}_{2C}$
\begin{align}\label{innn}
\| A_n G_n - p_n LG \|_{C,n}  \leq & \, \| n(T_n^{(0)} - 1\!\!1_n) G_n -
p_n L_0 G \|_{C,n}\\& + \| nT_n^{(1)} G_n - p_n L_1 G \|_{C,n}
+ \| nT_n^{(\geq2)} G_n \|_{C,n}.\notag
\end{align}
Note, that $p_n L_0 G =L_0 G_n$. Using the same arguments as in the
proof of Proposition~\ref{prop_ineq} we obtain
\begin{equation*}
\| n(T_n^{(0)} - 1\!\!1_n) G_n - p_n L_0 G \|_{C,n} + \|
nT_n^{(\geq2)} G_n \|_{C,n} \leq \frac{2}{n} \|G\|_{2C,n} \leq
\frac{2}{n} \|G\|_{2C}.
\end{equation*}
Next,
\begin{align*}
& \| nT_n^{(1)} G_n - p_n L_1 G \|_{C,n}\\ \leq &\, z\int_{\Gamma
_{\Lambda_n}}\sum_{\xi \subset \eta }\int_{{{\mathbb R}^d}}
\left\vert\left( 1-\frac{1}{n} \right) ^{| \xi | } 1\!\!1_{\Lambda
_n} (x) -1\right\vert |G\left( \xi \cup x\right)| \\ & \times
\prod\limits_{y\in \xi }e^{-\phi \left( y-x\right)
}\prod\limits_{y\in \eta \setminus \xi }\left( 1- e^{-\phi \left(
y-x\right) }\right) dx  C^{\left\vert \eta \right\vert }d\lambda
\left( \eta \right)\\ \leq &\, z\int_{\Gamma_{\Lambda_n}}\int_{\Gamma
_{\Lambda_n}} \int_{{{\mathbb R}^d}} \left[ 1- \left( 1-\frac{1}{n}
\right) ^{| \xi | } 1\!\!1_{\Lambda_n} (x) \right] |G\left( \xi \cup
x\right)| \\ & \times \prod\limits_{y\in \eta }\left( 1-
e^{-\phi \left( y-x\right) }\right) dx C^{\left\vert \eta \cup \xi
\right\vert }d\lambda \left( \eta \right)d\lambda \left( \xi
\right)\\ \leq & \, C\int_{\Gamma_{\Lambda_n}} \int_{{{\mathbb R}^d}}
\left[ 1- \left( 1-\frac{1}{n} \right) ^{| \xi | } 1\!\!1_{\Lambda
_n} (x) \right] |G\left( \xi \cup x\right)| dx C^{\left\vert \xi
\right\vert }d\lambda \left(
\xi \right),\\
\intertext{where we have used \eqref{integrability} and
\eqref{smallparam}. Using the same estimates as for \eqref{spec1} we
may continue} \leq &\, C\int_{\Gamma_{\Lambda_n}} \int_{\Lambda_n}
\left[ 1- \left( 1-\frac{1}{n} \right) ^{| \xi | }  \right] |G\left(
\xi \cup x\right)| dx C^{\left\vert \xi \right\vert }d\lambda \left(
\xi \right) \\& + C\int_{\Gamma_{\Lambda_n}} \int_{\Lambda
_n^c} |G\left( \xi \cup x\right)| dx
C^{\left\vert \xi \right\vert }d\lambda \left( \xi \right)\\
\leq &\, \frac{1}{n}\|G\|_{2C,n}+ C\int_{\Gamma_{0}} \int_{\Lambda
_n^c} |G\left( \xi \cup x\right)| dx C^{\left\vert \xi \right\vert
}d\lambda \left( \xi \right).
\end{align*}
But by the Lebesgue dominated convergence theorem,
\[
\int_{\Gamma_{0}} \int_{\Lambda_n^c} |G\left( \xi \cup x\right)| dx
C^{\left\vert \xi \right\vert }d\lambda \left( \xi \right)
\rightarrow 0, \quad n\rightarrow \infty.
\]
Indeed, $1\!\!1_{\Lambda_n^c} (x) |G\left( \xi \cup x\right)|
\rightarrow 0 $ point-wisely and may be estimated on $\Gamma
_0\times{{\mathbb R}^d}$ by $|G\left( \xi \cup x\right)|$ which is
integrable:
\[
C \int_{\Gamma_{0}} \int_{{{\mathbb R}^d}} |G\left( \xi \cup
x\right)| dx C^{\left\vert \xi \right\vert }d\lambda \left( \xi
\right) =\int_{\Gamma_0} |\xi| |G(\xi)| C^{\left\vert \xi
\right\vert }d\lambda \left( \xi \right) \leq \|G\|_{2C}<\infty.
\]
Therefore, by \eqref{innn}, the convergence \eqref{safcond} holds
for any $G\in{\mathcal L}_{2C}$, which completes the proof.
\end{proof}

\end{document}